\documentclass[a4paper,UKenglish,cleveref, autoref, thm-restate]{lipics-v2021}

\usepackage{wrapfig}

\usepackage[linesnumbered,lined,boxed,commentsnumbered]{algorithm2e}
\usepackage{multirow}
\usepackage{tabularx}

\title{Online 2-stage Stable Matching} 

\author{Evripidis {Bampis}}{Sorbonne Universit\'e, CNRS, LIP6, F-75005 Paris, France}{evripidis.bampis@lip6.fr}{}{}

\author{Bruno Escoffier}{Sorbonne Universit\'e, CNRS, LIP6, F-75005 Paris, France \and Institut Universitaire de France}{bruno.escoffier@lip6.fr}{}{
}

\author{Paul Youssef}{Univ. Grenoble Alpes, LIG, Grenoble, France}{paul.youssef@univ-grenoble-alpes.fr}{}{
}

\authorrunning{E. Bampis, B. Escoffier and P. Youssef} 

\Copyright{$ $} 



\category{} 

\relatedversion{} 

\acknowledgements{
}

\nolinenumbers 

\newtheorem{property}{Property}

\begin{document}

\maketitle

\begin{abstract}
We focus on an online 2-stage problem, motivated by the following situation: consider a system where students shall be assigned to universities. There is a  first round where some students apply, and a first (stable) matching $M_1$ has to be computed. However, some students may decide to leave the system (change their plan, go to a foreign university, or to some institution not in the system). Then, in a second round (after these deletions), we shall compute a second (final) stable matching $M_2$. As it is undesirable to change assignments, the goal is to minimize the number of divorces/modifications between the two stable matchings $M_1$ and $M_2$. Then, how should we choose $M_1$  and $M_2$?
We show that there is an {\it optimal online} algorithm to solve this problem. In particular, thanks to a dominance property, we show that we can optimally compute $M_1$  without knowing the students that will leave the system. We generalize the result to some other possible modifications in the input (students, open positions). 
 We also tackle the case of more stages, showing that no competitive (online) algorithm can be achieved for the considered problem as soon as there are 3 stages.
\end{abstract}


\section{Introduction}

Stable matchings have been extensively studied in the lite\-rature, both from a theoretical and a practical point of view.  
In the classical stable matching problem, one is given two equal-sized sets of agents, say men and women, where each person has strict preferences over the persons of the opposite sex. The goal is to match each man to exactly one woman and each woman to exactly one man, i.e., to find a perfect matching of men and women which is also {\em stable}. A perfect matching $M$ is {\em stable} if there is no {\em blocking pair}, i.e., a pair of a man and a woman who are not matched together in $M$, but they prefer each other more to their current partners in the matching. In 1962, Gale and Shapley, in their seminal paper~\cite{galeshapley62},  showed that a stable matching always exists, and designed a polynomial-time algorithm that finds such a matching. The stable matching problem is motivated by various applications where a centralized automated matching scheme is necessary in order to assign positions to applicants (matching of interns to hospitals \cite{M16}, \cite{MMT17}, university admission \cite{BB04}, school placement \cite{APR05}, faculty recruitment \cite{BB04}, etc.). 
In most of these applications, the matching schemes employ extensions of the Gale and Shapley algorithm taking into account particular ingredients of each application, including the use of incomplete preference lists, the existence of ties,  etc. 

Given the dynamic nature of many applications, there is an increasing interest on matching-related problems in the setting of dynamic graph algorithms where vertices or edges arrive or leave over time. A first work in this direction was proposed by Khuler et al. \cite{KhullerMV94} who considered the online stable marriage problem, where one is interested in the minimization of the number of blocking pairs. More recently, some studies are concerned with scenarios closely related to stable matchings, namely rank-maximal or (near) popular matchings  \cite{BHHK15}, \cite{GKP17}, \cite{NV19}. Biro et al., in \cite{BiroCF08}, studied the dynamics of stable marriage and stable roommates markets. Another interesting work in this setting is the one by Kanade, Leonardos and Magniez \cite{KLM16} who considered a setting where at each step, two random adjacent participants in some preference list are swapped and studied the problem of maintaining a matching while minimizing the number of blocking pairs.


A series of recent works tackle the situation where one wants to maintain stability of matchings when data evolves, while trying to minimize the modifications made in the matchings, as modifying pairs are usually highly non desirable in many applications:
\begin{itemize}
\item In \cite{GOS17a}, \cite{GOS17}, \cite{GSSO19}, Genc et al. study the notion of robustness in stable matching problems by introducing $(a,b)$-supermatches. An $(a,b)$-supermatch is a stable matching such that:
if $a$ pairs break up, a new stable matching can be found by changing the partners of these $a$ pairs and at most $b$ other pairs. They also define the most robust stable matching as one that requires the minimum number of repairs (i.e., minimizes $b$) among all stable matchings. They give some  complexity results and evaluate different heuristics using simulations.
\item In \cite{ChenSS21}, Chen et al. study the concepts of robustness where a matching must be stable even if the agents slightly change their preferences, and near stability where a matching must become stable if the agents slightly  adjust their preferences. They propose a polynomial-time algorithm that finds a socially optimal robust matching (if it exists), and they show  that the problem of finding a socially optimal and nearly stable matching is computationally hard.
\item In \cite{BredereckCKLN20}, Bredereck et al. study a 2-stage incremental version of the stable matching problem in terms of parametererized complexity. More precisely, one is given: a preference profile ${\cal P}_1$ for stage one, a preference profile ${\cal P}_2$ for stage two, a stable matching $M_1$ for profile ${\cal P}_1$
and a nonegative integer $k$. The question is whether there is a stable matching for stage two, $M_2$, whose distance from $M_1$ is smaller than or equal to $k$. They also study the incremental version of the stable roommates problem. They perform a parameterized complexity analysis for both problems with respect to the "degree of change" both in the input (preference profiles) and the output (stable matchings).
\item In \cite{GajulapalliLMV20}, Gajulapalli et al. considered stable assignment in different settings of the school choice problem. As in the previous work, they consider a 2-stage problem, but here only the instance $I_1$ is known at stage one. The instance $I_2$ becomes available only at stage two. The authors consider different variants where given an optimal solution for the instance of the first stage, they seek a stable assignment  of students to schools in two settings: In the first setting, it is disallowed to reassign the school of any student matched in stage one, and in the second setting the new stable assignment must provably minimize the number of such reassignments. Depending on the considered variant, they propose polynomial-time algorithms, or NP-hardness results. 
\end{itemize}

\noindent {\bf Our Contribution.}
This article lies in this line of research, combining stability requirements and low number of modifications in dynamic stable matching problems. The main difference is that most of these works adopt a reoptimization-like framework \cite{BHK18}, where {\it the first matching is fixed} and the  question is how to modify it by respecting some given constraints. In our case, we consider a 2-stage situation and we want to compute in an online manner a pair of solutions, one for each of the two stages, minimizing the number of modifications. 
Then, our main problem is how to choose the first stable matching without knowing the future so as to minimize the number of modifications in a 2-stage setting. Our approach is hence inspired by a new trend, the  online multistage optimization framework \cite{GTW14}, \cite{BEST21} and it is closely related to the 2-stage approach followed in  \cite{LeeS20} where a two-stage matching problem is considered in which the edges of the graph are revealed in two stages and in each stage the algorithm has to immediately and irrevocably extend the matching using the edges from that stage.
Furthermore, we note 
that several admission procedures do use two-rounds (or multi-rounds procedures), for instance this is the case for national college admissions in Sweden, in Turkey, (previously) in France, or for high school admissions in New-York city (see \cite{ADUE18}, \cite{HI21} and references therein).

  Here, we focus on a 2-stage problem, motivated by the following situation: consider a system where students shall be assigned to universities. There is a  first round where some students apply, and a first (stable) matching $M_1$ has to be computed. However, some students may decide to quit the system (change their plan, go to a foreign university, or to some institution not in the system). Then, in a second round (after these deletions), we shall compute a second stable matching $M_2$. The goal is to minimize the number of divorces/modifications between the two stable matchings $M_1$ and $M_2$. Then, how should we choose $M_1$ {\bf and} $M_2$? 

This problem will be called (2-L-SMP) (for 2-stage men Leaving Stable Matching Problem) and it is formally defined hereafter. We also consider the situation where new students arrive (2-A-SMP), and the case where there might be also some modifications in the open positions  (2-LA-SMP). 

We show that, quite surprisingly, there is an {\it optimal online} algorithm to solve these problems. In particular, thanks to a dominance property, we show that we can optimally compute $M_1$ {\it without knowing the students that will leave the system}. While we focus, for the sake of clarity, in the case of one-to-one (stable) matchings, we show that this result generalizes to the more general college-admission case. 
  We then tackle the case of more time steps, showing that no competitive (online) algorithm can be achieved for the considered problem as soon as there are 3 stages.

\noindent {\bf Organization of the article.}
 We give some definitions and formally define the considered problems in Section~\ref{sec:def}. In Section~\ref{sec:opt} we tackle one version of the problem and devise an online algorithm that we prove to be optimal. Section~\ref{sec:ext} shows how the result extends to the college-admission case and to the other 2-stage problems, and provides the negative result for the generalization to more stages.

\section{Definitions}\label{sec:def}

\subsection{Stable matching}

An instance $I$ of the stable matching problem involves two disjoint sets 
$U$ (of men) and $W$ (of women). Associated to each person is a strictly ordered preference list containing all the
members of the opposite sex. 

In all the article, a matching will denote a set of pairs $(u,w)\in U\times W$ such that any person (man or woman) is in at most one pair. If $(u,w)$ is a pair in a matching $M$, then $w$ is called the partner of $u$, and vice-versa.

\begin{definition}[blocking pair]
In a matching $M$, a {\it blocking pair} is a pair $(u,w)\not\in M$ such that both $u$ and $w$ would prefer to be partners than to be matched as in $M$. More precisely:
\begin{itemize}
    \item Either $u$ is not matched in $M$, or $u$  prefers $w$ to his partner $w'$ in $M$;
    \item and either $w$ is not matched in $M$, or $w$ prefers $u$ to her partner $u'$ in $M$.
\end{itemize}
\end{definition}


\begin{definition}[Stable matching] A matching is {\it stable} if it has no blocking pair.
\end{definition}

\begin{definition}[(Best) valid partners]
A {\it valid partner} of a person is a person 
 of the opposite sex such that there exist a stable matching in which
they are partners. 
The {\it best valid} partner of a person is his/her most preferred valid partner.
\end{definition}


Gale and Shapley’s fundamental result is that every instance of the stable matching problem admits at least
one stable matching~\cite{galeshapley62}. They proved this result by designing an efficient algorithm that is
guaranteed to find such a matching. Furthermore, Gale and Shapley showed that their algorithm finds a stable
matching with a nice property, namely that it gives all the men (or all the women, if the roles of the sexes
are reversed) simultaneously their best valid partner. This stable matching is called  men-optimal  (or women-optimal). 
The men-optimal  (resp. women-optimal)
stable matching has also the property that every woman (resp. man) has their worst valid partner for the
instance.

Another interesting and useful property of stable solutions is the Rural Hospitals Theorem~\cite{vitie70}. This theorem tackles the situation where $|U|\neq |W|$ - hence some people remained unmatched - and
states that if a person is not married in one stable matching then (s)he will not be in any other stable
matching. In other words, the set of matched persons is exactly the same in any stable matching.

A natural generalization of stable matching called the Maximum Weight Stable Matching (MWSM) is studied by Mai and Vazirani \cite{MV18}. Let $I$ denote an instance of the stable matching problem over sets $U$ and $W$ of men and women, respectively.
Let $f$ be a weight function over $U \times W$. The maximum weight stable
matching problem asks for a stable matching with maximum (total) weight. In~\cite{MV18} is given an efficient combinatorial algorithm for
it.
This generalization is also studied in the field of linear programming. The stable matching problem can be formulated with a linear
system. Vande Vate~\cite{VV89} showed that this linear program describes a polytope such that all its
extreme points are integral. Thus, solving the linear program solves the weighted stable marriage problem.

\subsection{Problem definition}

\begin{definition}
(2-A-SMP) In the {\it 2-stage women-arrival stable matching problem}, we are given:
\begin{itemize}
    \item A set $U$ of men, two sets $W_1$ and $W_2$ of women with $W_1\subseteq W_2$.
    \item Each men in  (resp. women) gives his (her) preferences (total ranking)  over $W_2$ (resp. over $U$).
\end{itemize}
The goal is to compute two matchings $(M_1,M_2)$ such that:
\begin{itemize}
    \item $M_1$ is stable for $(U,W_1)$ and $M_2$ is stable for $(U,W_2)$.
    \item The number of divorces $|M_1\setminus M_2|$ is minimized.
\end{itemize}
\end{definition}
We note that $|M_1\setminus M_2|$ counts the number of pairs that disappeared (divorces). As in a stable matching problem all the stable matchings have the same size (as everyone prefers to be matched than unmatched), we could equivalently maximize $|M_1\cap M_2|$, or minimize the set of new pairs $|M_2\setminus M_1|$. 

We are interested in the online version of the problem where we have to compute $M_1$ at stage 1 while having no knowledge about $W_2$. In other words, at stage 1, we only know $U$, $W_1$, and the preferences between men in $U$ and women in $W_1$. We note that these preferences between $U$ and $W_1$ do not change between the two stages.\\

\noindent \textbf{Variants of the problem}

\begin{itemize}
    \item (2-L-SMP) 2-stage Men-Leaving Stable Matching Problem. Here, the set of women is fixed, but men are leaving: we have two sets $U_1,U_2$ with $U_2\subseteq U_1$.
    \item (2-LA-SMP), where both men are leaving, and women are arriving: we have $(U_1,W_1)$ and $(U_2,W_2)$ with $U_2\subseteq U_1$ and $W_1\subseteq W_2$.
    \item ($T$-A-SMP), ($T$-L-SMP), ($T$-LA-SMP): these are the generalizations of the previous problems where we have $T$ stages instead of 2, and the goal is to compute $T$ stable matchings $M_1,\dots,M_T$, where the number of divorces $\sum_{t=1}^{T-1}|M_i\setminus M_{i+1}|$ is minimized. For instance, in $T$-A-SMP, the set $U$ of men is fixed, and we have $T$ sets $W_1\subseteq W_2\subseteq \dots \subseteq W_T$ of women. At stage $t$, we have to compute $M_t$ while having no information on the women that will arrive in the future.
\end{itemize}

We note that we do not consider the most general case with arrivals and departures on both sides, as simple examples show that there is no (constant) competitive algorithm even for 2 stages in this case.

\section{An optimal online algorithm for {2-A-SMP}}\label{sec:opt}

We will denote $\Delta_1$ the set of stable matchings in stage 1 (with $U$ and $W_1$) and $\Delta_2$ the set of stable matchings in stage 2 (with $U$ and $W_2$).


Our algorithm mainly relies on a dominance property, which is the main technical result of this work, that allows to make an optimal choice at the first stage.

\begin{definition}
Let two stable matchings $M$ and $M'$, in a given stable matching problem over men set $U$ and women set $W$. We say that $M$ men-dominates $M'$ if for any man $u\in U$, his partner in $M$ is at least as good (according to his preferences) as his partner in $M'$.
\end{definition}

\begin{property}[Dominance property]\label{lemma:dominance}
Let $(M_1,M_2)\in \Delta_1\times \Delta_2$. Let $M'_1\in \Delta_1$ such that $M'_1$ men-dominates $M_1$. Then there exists $M'_2\in \Delta_2$ such that $|M'_1\setminus M'_2| \leq |M_1\setminus M_2|$.
\end{property}

This property (proved later on) says that we shall always prefer at stage 1 a matching $M'_1$ that men-dominates another matching $M'_2$, {\it whatever the set of women $W_2\setminus W_1$ that will arrive in the second stage, and the preferences between $U$ and $W_2\setminus W_1$}.

Based on this property, we can consider the following algorithm {\it Opt-2-Stage}:
\begin{enumerate}
    \item At stage 1, compute the men-optimal stable matching $M^*_1\in \Delta_1$.
    \item At stage 2, compute a maximal weight stable matching \cite{MV18} $M^*_2\in \Delta_2$ where the weight of $(u,w)$ is 1 if $(u,w)$ is in $M^*_1$, and 0 otherwise.
\end{enumerate}

\begin{theorem}
{\it Opt-2-Stage} is a polynomial time online algorithm which outputs an optimal solution of {2-A-SMP}.  
\end{theorem}
\begin{proof}
As computing a men-optimal stable matching, and a maximum weight stable matching, can be done in polynomial time, {\it Opt-2-Stage} runs in polynomial time.

Let $(M^*_1,M^*_2)$ be the solution output by {\it Opt-2-Stage}, and $(M_1,M_2)$ be an optimal solution. 

As $M^*_1$ men-dominates $M_1$, by Property~\ref{lemma:dominance} there exists $M'_2\in \Delta_2$ such that $|M^*_1\setminus M'_2|\leq |M_1\setminus M_2|$.

In the second stage, as we put weight 1 for each edge in $M^*_1$, we have that $|M^*_2\cap M^*_1|\geq |M'_2\cap M^*_1|$. Consequently,  $|M^*_1\setminus M^*_2|\leq |M^*_1\setminus M'_2|$.
Finally, $|M^*_1\setminus M^*_2|\leq |M_1\setminus M_2|$ and $(M^*_1,M^*_2)$ is optimal.
\end{proof}

In the remainder of this Section we prove the dominance property. To do so, we will need the following notion of difference graph between two matchings.

\begin{definition}[Difference graph]
Let $M$ and $M'$ be two matchings on an instance of the SMP with men set $U$ and women set $W$. The difference graph $G(M,M')$ has vertex set $U\cup W$, and edge set $M\triangle M'=(M\setminus M')\cup (M\setminus M')$. 
\end{definition}
Note that as $M$ and $M'$ are matchings, $G(M,M')$ is composed of isolated vertices and (vertex disjoint) cycles and paths, which alternate edges from $M$ and from $M'$.

\subsection{Preliminary properties}

\subsubsection{Coherence of matched persons}

When a woman arrives between stages one  and two, intuitively the ``competition'' gets harder for women, and easier for men. So in particular if a man is matched in $M_1$, he should be matched as well in $M_2$, and if a woman is not matched in $M_1$ she should be unmatched in $M_2$ as well. We prove this property\footnote{The fact that matched men remained matched is proven in Lemma 12 of \cite{GajulapalliLMV20}, we re-prove it here for completeness.}, which will be useful for simplifying the analysis later. 

\begin{property}\label{prop:matched}
Let $M_1\in \Delta_1$, $M_2\in \Delta_2$. If a man is matched in $M_1$ then he is matched in $M_2$. If a woman is unmatched in $M_1$, then she is unmatched in $M_2$.
\end{property}
\begin{proof}
Let $u_1\in U$ matched in $M_1$ with $w_1$. Suppose that $u_1$ is not matched in $M_2$, and consider the difference graph $G(M_1,M_2)$ (on $U\cup W_2$). As $u_1$ has degree 1, it is the endpoint of a chain $C$. Every woman is matched in $M_2$ (there cannot be both a man and a woman unmatched), so $C$ has an even number of edges, $C=(u_1,w_1,u_2,w_2,\dots,u_k)$ with $(u_i,w_i)\in M_1$ and $(u_{i+1},w_i)$ in $M_2$. Note that $u_k$ is not matched in $M_1$.

As $u_1$ is unmatched in $M_2$, $w_1$ prefers $u_2$ to $u_1$, otherwise $(u_1,w_1)$ would be a blocking pair for $M_2$. Then $u_2$ prefers $w_2$ to $w_1$, otherwise $(u_2,w_1)$ would be a blocking pair for $M_1$. Consequently, $w_2$ prefers $u_3$ to $u_2$, otherwise $(u_2,w_2)$ would be a blocking pair for $M_2$. With an easy recurrence, we get that $w_i$ prefers $u_{i+1}$ to $u_i$ all along the chain. Then $w_{k-1}$ prefers $u_k$ to $u_{k-1}$, and $(u_k,w_{k-1})$ is a blocking pair of $M_1$, contradiction.

Similarly, suppose that  there is a woman $w_1$ matched in $M_2$ but not in $M_1$. In $G(M_1,M_2)$ $w_1$ is a the endpoint of a chain $C$. The other endpoint of $C$ cannot be a man, as any man is matched in $M_1$ (since $w_1$ is not matched in $M_1$, and the set of men did not change). So there is a chain $C=(w_1,u_1,w_2,u_2,\dots,w_k)$ with $(w_i,u_i)\in M_2$ and $(u_{i},w_{i+1})$ in $M_1$. Then the same argument as before applies, leading to a blocking pair.
\end{proof}

\subsubsection{Reduction to regular instances}

Let us call an instance {\it regular} if $|W_1|\leq |W_2|=|U|$. We first show that we can restrict w.l.o.g. to regular instances, which will substantially simplify the case analysis for proving the dominance lemma.

\begin{lemma}
The dominance property (Property~\ref{lemma:dominance}) is true if and only if it is true on regular instances.
\end{lemma}
\begin{proof}
 Suppose first that $|U|<|W_1|$. Then there is a set $W_0$ of unmatched women at stage one. By the rural-hospital theorem, this set is the same for all stable matchings at stage 1. As, by Property~\ref{prop:matched}, these women are not matched at stage 2 as well, we can safely remove $W_0$ at both stages: this does not change the set of stable matchings at each stage. After doing this, we can assume 
 that $|W_1|\leq |U|$.

Now, suppose that $|W_2| > |U|$. Note that by assuming $|U|\geq |W_1|$ (thanks to the previous case) all women are matched in the first stage. Let $k=|W_2|-|U|$, and $W'=\{w'_1,\dots,w'_k\}$ be the set of unmatched women at step 2 (which are the same in any stable matching, by the rural-hospital theorem). We add a set $T=\{t_1,\dots,t_k\}$ of $k$ new men at both steps, with the following preferences:
\begin{itemize}
    \item They are ranked in the last $k$ positions of all women;
    \item Among the men in $T$, woman $w'_i$ prefers $t_i$.
    \item Each $t_i$ prefers any woman in $W'$ to any woman in $W\setminus W'$.
    \item Among the women in $W'$, $t_i$ prefers $w'_i$.
\end{itemize}

Then clearly: 
\begin{itemize}
    \item none of these new men is matched at step 1;
    \item at step 2, in any stable matching $t_i$ is matched with $w'_i$.
\end{itemize}
So, adding this set $T$ of new men (1) does not modify the set of stable matchings at step 1 and (2) add $k$ {\it new} pairs $(t_i,w'_i)$ to any stable matching at step 2. So, back to the dominance property, this modification does not modify $|M_1\setminus M_2|$ in a pair $(M_1,M_2)$ of stable matchings.


Finally, if $|W_2|<|U|$, then let $U'=\{u'_1,\dots,u'_k\}$ be the set of unmatched men at stage 2, where $k=|U|-|W_2|$. By Property~\ref{prop:matched}, these men are also unmatched at stage 1. We can add to the instance a set $W'=\{w'_1,\dots , w'_k\}$ of  dummy women at stage 2, while enforcing a perfect matching between $U'$ and $W'$ in any stable matching. This can be easily done by setting $w'_k$ to be the most preferred woman of $u'_k$ and $u'_k$ the most preferred man of $w'_k$. This modification adds the same $k$ pairs to any stable matching $M_2$ at the second stage, and does not change the number of divorces between $M_2$ and any stable matching $M_1$ at stage 1. 
\end{proof}

From now, we restrict ourselves to regular instances. In particular, any stable matching at stage 2 is a perfect matching.

\subsubsection{Paths and cycles in the difference graph}

\begin{lemma}\label{lemme:structure}
Let $M_1\in \Delta_1$ and $M_2\in \Delta_2$. 
In the difference graph $G(M_1,M_2)$ (on $U\cup W_2$), in each cycle $C$:
\begin{itemize}
    \item (Type I) Either each man in $C$ strictly prefers his partner in $M_2$ to his partner in $M_1$, and each woman strictly prefers her partner in $M_1$ to her partner in $M_2$;
    \item (Type II) Or each man in $C$ strictly prefers his partner in $M_1$ to his partner in $M_2$, and each woman strictly prefers her partner in $M_2$ to her partner in $M_1$.
\end{itemize}
Moreover, for each chain $P$:
\begin{itemize}
    \item The extremal edges both belong to $M_2$.
    \item Concerning the internal vertices in $P$ (which are matched in both matchings), each man strictly prefers his partner in $M_2$ to his partner in $M_1$, and each woman strictly prefers her partner in $M_1$ to her partner in $M_2$.
\end{itemize}
\end{lemma}
\begin{proof}
Let us first consider a cycle $C=(w_1,u_1,w_2,u_2,\dots,w_k,u_k,w_1)$, where edges $(u_i,w_i)$ are from $M_1$ and the other edges from $M_2$.  

Suppose that $u_1$ prefers $w_2$ to $w_1$. Then necessarily: 
\begin{itemize}
    \item $w_2$ prefers $u_2$ to $u_1$, otherwise $(u_1,w_2)$ would have been a blocking pair for $M_1$.
    \item $u_2$ prefers $w_3$ to $w_2$, otherwise $(u_2,w_2)$ would have been a blocking pair for $M_2$.  
\end{itemize}
By an easy recurrence, we obtain that  $w_j$ prefers $u_j$ to $u_{j-1}$ for any $j=2,\dots,T$ (and $w_1$ prefers $u_1$ to $u_k$), and $u_j$ prefers $w_{j+1}$ to $w_j$. This corresponds to cycles of type I.

If $u_1$ prefers $w_1$ to $w_2$, then with a symmetric argument we obtain a cycle of type II.\\

Let us now consider a chain $P$. Thanks to the assumption that the instance is regular, $M_2$ is a perfect matching so the extremal edges must belong to $M_2$. So $P$ has an odd number of edges, and can be written as $(u_1,w_1,\dots,u_k,w_k)$, with $(u_i,w_i)\in M_2$. As $M_1$ is stable, $w_1$ prefers $u_2$ to $u_1$. Then, $u_2$ prefers $w_2$ to $w_1$, otherwise $(u_2,w_1)$ would be blocking for $M_2$. Then, by stability of $M_1$, $w_2$ prefers $u_3$ to $u_2$. By recurrence, we get that all along the chain $u_i$ ($i\geq 2$) prefers $w_i$ (his partner in $M_2$) to $w_{i-1}$ (his partner in $M_1$), and $w_i$ ($i\leq k-1)$ prefers $u_{i+1}$ to $u_i$.
\end{proof}

\subsection{Proof of the dominance property}

Now, we are able to prove the dominance property. We recall that we consider a regular instance.

\noindent\textbf{Construction of $M'_2$.} Let $M_1\in \Delta_1$, $M_2\in \Delta_2$, and $M'_1\in \Delta_1$ where $M'_1$ man-dominates $M_1$. We build $M'_2$ as follows, from the two matchings $M'_1$ and $M_2$. First, we put in $M'_2$ the set of edges $M'_1\cap M_2$ where $M'_1$ and $M_2$ agree. Then we consider the difference graph $G(M'_1,M_2)$. In this graph:
\begin{itemize}
\item For each path, we take in $M'_2$ the edges of $M_2$.
\item For each cycle  of Type I, we take in $M'_2$ the edges of $M_2$.
\item For each cycle of Type II, we take in $M'_2$ the edges of $M'_1$.
\end{itemize}

Note that $M'_2$ is a perfect matching.

We now prove that $|M'_1\setminus M'_2|\leq |M_1\setminus M_2|$ (Lemma~\ref{lemma:good}), and that $M'_2$ is stable (Lemma~\ref{lemma:stable}), which concludes the proof of the dominance property. 

\begin{lemma}\label{lemma:good}
$|M'_1\setminus M'_2|\leq |M_1\setminus M_2|$.
\end{lemma}
\begin{proof}

Let $u$ be a man who divorced between $M'_1$ and $M'_2$ - so $u$ is matched in $M'_1$, hence in $M'_2$ by Property~\ref{prop:matched} but with a different person. Then $u$ belongs to a path or to a cycle of Type I in $G(M'_1,M_2)$. Note that if he belongs to a path it is an internal vertex as he is matched in both matchings. Then, following Lemma~\ref{lemme:structure} (applied with $M'_1$ and $M_2$), in both cases $u$ (strictly) prefers his partner in $M_2$ to his partner in $M'_1$. As $M'_1$ men-dominates $M_1$, $u$ strictly prefers his partner in $M_2$ than in $M_1$. This means that $u$ also got divorced between $M_1$ and $M_2$.

So the set of divorced men between $M'_1$ and $M'_2$ is included in the set of divorced men between $M_1$ and $M_2$.
\end{proof}

\begin{lemma}\label{lemma:stable}
$M'_2$ is stable.
\end{lemma}
\begin{proof}
Suppose that there is a blocking pair $(u,w')$ in $M'_2$. Let $w$ be the partner of $u$ and $u'$ be the partner of $w'$ in $M'_2$.

We cannot have both $(u,w)$ and $(u',w')$  in $M_2\cap M'_2$, as $(u,w')$ would be blocking for $M_2$ which is stable. Also, we cannot have both $(u,w)$ and $(u',w')$  in $M'_1\cap M'_2$ as $(u,w')$ would be blocking for $M'_1$ which is stable.

Note that $M'_2\subseteq M'_1\cup M_2$. So we are left with two possible cases:
\begin{itemize}
    \item Case 1: $(u,w)\in M'_1\setminus M_2$. Then $(u',w')\in M_2\setminus M'_1$. By construction, $(u,w)$ is in a cycle of Type II (in $G(M'_1,M_2)$) and  $(u',w')$ is in a cycle of Type I or in a path.
    \item Case 2 (vice-versa): $(u,w)\in M_2\setminus M'_1$. Then $(u',w')\in M'_1\setminus M_2$. By construction, $(u,w)$ is in a cycle of Type I or in a path, and  $(u',w')$ is in a cycle of Type II.
\end{itemize}

In the first case, by Lemma~\ref{lemme:structure},    $u$ prefers his partner in $M'_1$, i.e., $w$, to his partner in $M_2$.
If $(u,w')$ were blocking for $M'_2$, $u$ would prefer $w'$ to $w$, so he would prefer $w'$ to his partner in $M_2$. Also, $w'$ prefers $u$ to $u'$ which is her partner in $M_2$. Hence, $(u,w')$ would be blocking for $M_2$.

In the latter case, note that $u$ is matched in $M'_1$ otherwise $(u,v')$ would be blocking for $M'_1$. Then $u$ is either in a cycle of Type I or an internal vertex of a path. By Lemma~\ref{lemme:structure},
    $u$ prefers his partner in $M_2$, i.e., $w$, to his partner in $M'_1$.
If $(u,w')$ were blocking for $M'_2$, $u$ would prefer $w'$ to $w$, so he would prefer $w'$ to his partner in $M'_1$. Also, $w'$ prefers $u$ to $u'$ which is her partner in $M'_1$. Hence, $(u,w')$ would be blocking for $M'_1$.
\end{proof}

\section{Extensions}\label{sec:ext}

\subsection{University-admission case}

In the university-admission case, each $u_i\in U$ (now university) is given with a (positive integer) capacity $c_i$. Then, at most $c_i$ elements of $W$ (now students) can be assigned to the university $u_i$. Notions of stability, and classical results for stable matchings (Gale-Shapley algorithms, rural hospitals theorem,\dots), are well known to generalize to this more general setting. A way to see this is to transform an instance of the university-admission case to a standard stable matching problem as follows:
\begin{itemize}
    \item Each university $u_i$ with capacity $c_i$ is transformed into $c_i$ elements $u_i^j$, $j=1,\dots,c_i$, where each $u_i^j$ has the same preference list as $u_i$.
    \item For $w_j\in W$, we transform her preference list by replacing $u_i$ by the sequence $u_i^1 \dots u_i^{c_i}$. For instance, if the preference list of $w_j$ starts with $(u_2,u_4,\dots)$ where $u_2$ and $u_4$ have capacity 2, then it becomes $(u_2^1,u_2^2,u_4^1,u_4^2,\dots)$.
\end{itemize}
Then any stable matching in the transformed instance corresponds to a stable assignment in the initial university-admission instance.

In our two stage problem, when dealing with university-admission, we want to find a pair $(A_1,A_2)$ of assignments ($A_i$ is an assignment of students to universities at stage $i$). The goal is then to minimize the number of assignments that have been modified between the two stages (i.e., the number of students whose university has changed between the two stages). 

We note that there is a difficulty to which we have to pay attention. The number of modifications in the assignments {\it does not} correspond to the number of modifications in the transformed stable matching instance. Indeed, in this transformed stable matching instance, if $w_i$ is matched to $u_j^1$ in the first stage and to $u_j^2$ in the second stage, this corresponds to a modification in the matchings, but in both stages student $w_i$ is assigned to university $u_j$, so this is {\it not} a modification in the student-university assignment. 

However, our results extend to this general university-admission case. We sketch the proof here.
\begin{itemize}
    \item  The algorithm {\it Opt-2-stage} easily generalizes. We consider the transformed stable matching instance; stage 1 remains unchanged (we compute a man-optimal stable matching $M^*_1$). Then, if $w_i$ is assigned to $u_j^k$ in $M^*_1$, at stage 2 we put a weight 1 on all pairs $(u_j^\ell,w_i)$, $\ell=1,\dots,c_j$, as these pairs correspond to the same assignment of student $w_i$ to university $u_j$.

\item Dominance property. 
The key element that remains to be checked in the dominance property, restated in terms of modifications of assignments.  
We use the transformed (stable matching) instance, with the very same procedure to compute $M'_2$. 
It is easy to see that Property 2 and Lemma 1 are still valid, as they only rely on the fact that a person is matched or not. Also, Lemma 2 and Lemma 4 are still valid, as they deal with stability properties of the matchings. 

The central point is Lemma 3 as, as said before, the number of modifications that we count  is no more $|M1\setminus M_2|$ but the number of modifications in the underlying student-university assignments. 

Let us prove that Lemma 3 still holds. Consider that one student $w_i$ is assigned to (an occurrence of) university $u_j$ in $M'_2$ but was not assigned to (an occurrence of) university $u_j$ in $M'_1$. By property 2, $w_i$ was matched in $M'_1$, i.e., assigned to a university $u_k$ (with $u_k\neq u_j$). $w_i$ is then an internal  vertex on a path or a cycle of Type 1 in $G(M'_1,M_2)$. Then by Lemma 2 $w_j$ strictly prefers her partner in $M'_1$ to her partner in $M'_2$. Note that by construction of the preferences in the transformed instance, this means that $w_j$ strictly prefers university $u_k$ to university $u_j$. As $M'_1$ men-dominates $M_1$, $M_1$ women-dominates $M'_1$, and $w$ prefers her partner in $M_1$ to her partner in $M'_1$. In term of university assignment, $w_i$ was assigned in $M_1$ to a university at least as good (for her) as her university $u_k$ in $M'_1$. In $M_2$ she got a university $u_j$ which is strictly worse than $u_k$. So she also changed university between $M_1$ and $M_2$.  

In other words, here again, if a student got a new (different) assignment/university in $M'_2$ with respect to $M'_1$, she also got a new (different) assignment/university in $M_2$ with respect to $M_1$.
\end{itemize}

\subsection{When  men are (also) leaving}

In Section~\ref{sec:opt} we tackled the case where the set of men was fixed, and some women arrived between stages 1 and 2.

Let us now consider an instance $I$ of (2-L-SMP) where men may leave the game, i.e., the set of women $W$ is fixed, and the men set is $U_1$ at time 1, and $U_2$ at time 2 with $U_2\subseteq U_1$. Let $\overline{U}=U_1\setminus U_2=\{\overline{u}_1,\dots,\overline{u}_k\}.$

We build the following instance $I'$ of  (2-A-SMP):
\begin{itemize}
    \item The set of men is $U'=U_1$.
    \item The set of women is $W'_1=W$ at stage 1, and $W'_2=W\cup \overline{W}$, where $\overline{W}=\{\overline{w}_1,\dots,\overline{w}_k\}$.
    \item The most preferred partner of $\overline{u}_i$ is $\overline{w}_i$, and the most preferred partner of $\overline{w}_i$ is $\overline{u}_i$. 
\end{itemize}
Let $\overline{M}$ be the set of pairs $\{(\overline{u}_i,\overline{w}_i),i=1,\dots,k\}$.
\begin{lemma}
$M$ is a stable matching in the second stage of $I$ if and only if $M_2\cup \overline{M}$ is a stable matching in the second stage  of $I'$.
\end{lemma}
\begin{proof}
This easily follows from the fact that any matching in the second stage of $I'$ contains $\overline{M}$.
\end{proof}
So, $(M_1,M_2)\rightarrow (M_1,M_2\cup \overline{M})$ is a one-to-one correspondence between the 
sequence of stable matchings in $I$ and $I'$. The number of divorces is precisely the same (as men matched in $\overline{M}$ are not in the second stage of $I$).\\  

The very same argument works also for the problem (2-LA-SMP) where both men leave and women arrive between the two stages (proof omitted). Hence, the following holds.

\begin{theorem}
{\it Opt-2-Stage} is a polynomial time online algorithm which outputs an optimal solution of {2-LA-SMP} (and {2-L-SMP}).  
\end{theorem}

\subsection{No competitive algorithm for more stages}

A natural extension is to consider the problem on a larger number of stages. With more than two stages, is it still possible to find an optimal online algorithm? or at least a competitive online algorithm? 

We answer negatively to this question, by showing that it is not the case already for 3 stages.

\begin{theorem}
For 3 stages, for any $c$, there is no online $c$-competitive (deterministic) algorithm for 3-A-SMP.
\end{theorem}
\begin{proof} 
We build an instance with 3 stages, a unique stable matching at time 1, two stable matchings at time 2. The third stage depends on the choice of the algorithm at time 2.

As a building block in our construction, let us consider the following instance with $n-1$ men and women, and the following cyclic preferences:

\medskip

\begin{center}
\begin{tabular}{|l|l|}
\hline
{\bf Men} & {\bf Women} \\
\hline
$u_1: w_1 w_2 \dots w_{n-1}$ & $w_1: u_2 u_1 u_{n-1} \dots u_3$ \\
$u_2: w_2 w_3 \dots w_1$ & $w_2: u_3 u_2 \dots u_4$ \\
$\dots$ & $\dots$ \\
$u_{n-1}:w_{n-1} w_1 \dots w_{n-2}$ & $w_{n-1}:w_{1} w_{n-1} w_{n-2} \dots w_{2}$ \\
\hline
\end{tabular}
\end{center}

\medskip

It is not hard to see that there are only 2 stable matchings: the men optimal $M^h_2$ made of $(u_i,w_i)$ for all $i$, and the women optimal $M^f_2$ made of $(u_{i+1},w_i)$.\\ 

Now we can describe the first two stages of the instance, (the 3-stage instance contains $n$ men and respectively 1, $n-1$, and $n$ women at each of the 3 stages). 
\begin{itemize}
                \item At stage 2, we have $n$ men and $n-1$ women. The preferences are as in the previous instance, plus $u_{n}$ with preferences $w_1 w_2 \dots w_n$. $u_n$ is ranked last by every woman, $w_i,i=1,\dots,n-1$. As $u_n$ is in no stable matching, there are only 2 stable matchings, $M^h_2$ and $M^f_2$. 
                \item At stage 1, there is only the woman $w_1$, and then only one stable matching $M_1=(u_2,w_1)$.
\end{itemize}

At stage 1 the algorithm has no choice as there is a unique stable matching. At stage 2 it can choose either $M^h_2$ or $M^f_2$.

{\bf Case 1.} Suppose that it chooses $M^h_2$. Then it makes one divorce between stages 1 and 2 (pair $(u_2,w_1)$), while $M^f_2$ makes no divorce. We give an instance at stage 3 where we can maintain all the pairs in $M^f_2$. 

To do this, at stage 3 where woman $w_n$ arrives, we put $u_{n}$ in first position in the ranking of $w_n$: then $M^f_2$ plus the pair $(u_n,w_n)$ is stable (all the women have their first choice). So there is a solution with no divorce (with value 0), while the algorithm made at least one divorce.

{\bf Case 2.} Suppose that it chooses $M^f_2$.
Then we force the algorithm to change everything at stage 3, where there will be a unique stable matching, the men optimal one.

To do this, at stage 3 where woman $w_n$ arrives, we put $w_{n}$ in second position in the ranking of men $u_1,\dots,u_{n-1}$. $w_n$ is in last position for $u_n$. The preference of $w_n$ is $u_1 u_2 \dots u_{n}$. 

Note that the partner of $u_{n}$ is  $w_n$ in all stable matchings. Indeed, suppose that his partner is  $w_i,i<n$. $u_n$ is ranked last by $w_i$, so $w_i$ prefers $u_i$, and $w_i$ is the first choice of $u_i$, so $(u_i,w_i)$ is a  blocking pair, contradiction.

Then for any $i<n$ the partner of $u_i$ is $w_i$ in any stable matching. Indeed, if $u_i$ were matched with $w_j,j\neq i$, then $u_i$ would prefer $w_{n}$ to $w_j$, and $w_n$ prefers anyone to her husband $u_n$, so $(u_i,w_n)$ would be a blocking pair.

So the matching $(u_i,w_i)$ for all $i$ is the unique stable matching at stage 3. The algorithm makes $n-1$ divorces (between stages 2 and 3), while taking $M^h_2$ at stage 2 allows to make only 1 divorce in total (between stage 1 and 2, with no divorce between stage 2 and 3).
\end{proof}
As a remark, this example actually shows that no $(n-1-\epsilon)$-competitive algorithm exists.

We note however that, interestingly, the dominance condition still holds. Indeed, we have the following result. Let $\Delta_t$ be the set of stable matchings at stage $t$.

\begin{lemma}
Let $(M_1,M_2,\dots,M_T)\in \prod_{t=1}^T \Delta_t$. Let $M'_1\in \Delta_1$ such that $M'_1$ men-dominates $M_1$. Then there exists $(M'_2,\dots,M'_T)\in \prod_{t=2}^T \Delta_t$ such that $\sum_t |M'_t\setminus M'_{t+1}| \leq \sum_t |M_t\setminus M_{t+1}|$.
\end{lemma}
\begin{proof}
Starting from $M'_1$, we build as in Lemma~\ref{lemma:dominance} a stable matching $M'_2\in \Delta_2$ such that $|M'_2\setminus M'_{1}| \leq  |M_2\setminus M_{1}|$.
As it can be seen in the proof of Lemma~\ref{lemma:dominance}, $M'_2$ men-dominates $M_2$. So we can apply again Lemma~\ref{lemma:dominance} to build a stable matching $M'_3\in \Delta_3$ such that $|M'_3\setminus M'_{2}| \leq  |M_3\setminus M_{2}|$. By an easy recurrence we build the sequence $M'_t,t=2,\dots,T$.
\end{proof}

As a corollary, as the men-optimal matching dominates all other stable matchings, we get the following. 

\begin{corollary}
There always exists an optimal solution that chooses the men-optimal matching at stage 1.
\end{corollary}

\section{Conclusion}
We showed in this article that the considered 2-stage stable matching problems admit an optimal online algorithm. While such an optimal online algorithm does not exist for more than 2 stages in the considered model, studying stable matching problems on more stages seems to be an interesting research direction. For instance, we can think of using randomized online algorithms to reach (asymptotic) competitive ratios, or make further assumptions on the model -- for instance in several online matching problems people arrive one by one in the game. The study of the off-line problem could be also of interest, as well as extensions of the results to a more general preference model (with ties, incomplete preferences,\dots).

\bibliography{biblio}

\end{document}